\documentclass[letterpaper, 10 pt, conference]{ieeeconf} 
\IEEEoverridecommandlockouts                              
\usepackage{hyperref}
\usepackage{cite}
\usepackage{amsmath,amssymb,amsfonts,amsbsy}
\usepackage{algorithm}
\usepackage{algorithmic}
\usepackage{graphicx}
\usepackage{comment}
\usepackage{theorem,url}
\usepackage{tabularx} 
\usepackage{bm}
\usepackage{dsfont}
\usepackage{float}
\usepackage{xcolor}

{\theorembodyfont{\slshape}\newtheorem{theorem}{Theorem}[section]}
{\theorembodyfont{\slshape}}
{\theorembodyfont{\slshape}\newtheorem{lemma}[theorem]{Lemma}}
{\theorembodyfont{\slshape}}
{\theorembodyfont{\slshape}}
{\theorembodyfont{\upshape}\newtheorem{definition}[theorem]{Definition}}
{\theorembodyfont{\upshape}}
{\theorembodyfont{\upshape}\newtheorem{remark}[theorem]{Remark}}
{\theorembodyfont{\upshape}\newtheorem{assumption}[theorem]{Assumption}}
{\theorembodyfont{\upshape}\newtheorem{example}[theorem]{Example}}

\newcommand{\setR}{\mathbb{R}}

\newcommand{\setX}{\mathbb{X}}

\newcommand{\nonnegR}{\mathbb{R}_{\geq0}}
\newcommand{\posR}{\mathbb{R}_{>0}}

\newcommand{\setRnx}{\mathbb{R}^{n_x}}
\newcommand{\setRny}{\mathbb{R}^{n_y}}

\newcommand{\graphA}{\mathcal{A}}
\newcommand{\graphAbar}{\bar{\mathcal{A}}}
\newcommand{\graphAtree}{\mathcal{A}_{\text{tree}}}

\newcommand{\graphC}{\mathcal{C}}

\newcommand{\graphE}{\mathcal{E}}

\newcommand{\graphG}{\mathcal{G}}

\newcommand{\graphP}{\mathcal{P}}

\newcommand{\graphV}{\mathcal{V}}
\newcommand{\graphW}{\mathcal{W}}

\newcommand{\graphX}{\mathcal{X}}

\newcommand{\graphAC}{\graphG([\graphA^\top,\graphC^\top])}
\newcommand{\graphACbar}{\graphG([\graphAbar^\top,\graphC^\top])}

\newcommand{\Aadj}{A_{\mathrm{adj}}}
\newcommand{\graphVext}{\mathcal{V}_{\mathrm{e}}}
\newcommand{\graphVint}{\mathcal{V}_{\mathrm{i}}}

\newcommand{\Ainc}{A_{\mathrm{inc}}}

\newcommand{\strucset}{\{0, \ast, ?\}}
\newcommand{\MatrinStruc}[1]{#1\!\in\! \mathcal{P}(\mathcal{#1})}

\title{\LARGE \bf Scalable Sensor Placement for Cyclic Networks with Observability Guarantees: Application to Water Distribution Networks}

\author{J.J.H. van Gemert, V. Breschi, D.R. Yntema, K.J. Keesman, M. Lazar
\thanks{This work was performed in the corporation framework of Wetsus, European Centre of Excellence for Sustainable Water Technology (www.wetsus.nl). Wetsus is co-funded by the Dutch Government (Ministry of Economic Affairs and Climate Policy, Ministry of Education, Culture and Science and Ministry of Infrastructure and Water Management) and the Province of Fryslan. The authors like to thank the participants of the research theme `Smart Water Grids' for the fruitful discussions and their financial support.}
\thanks{J.J.H. van Gemert, V. Breschi and M. Lazar are with the Department of Electrical Engineering, Eindhoven University of Technology, The Netherlands. E-mail of corresponding author: {\tt\small J.J.H.v.Gemert@tue.nl}}
\thanks{D. Yntema and K.J. Keesman are with Wetsus, Centre of Excellence for Sustainable Water Technology, 8900 Leeuwarden, The Netherlands. }
\thanks{
K.J. Keesman is also with Mathematical and Statistical Methods – Biometris, Wageningen University, Wageningen, The Netherlands.} 
}

\begin{document}
\maketitle
\thispagestyle{empty}
\pagestyle{empty}
\begin{abstract}
Optimal sensor placement is essential for state estimation and effective network monitoring. As known in the literature, this problem becomes particularly challenging in large-scale undirected or bidirected cyclic networks with parametric uncertainties, such as water distribution networks (WDNs), where pipe resistance and demand patterns are often unknown. 
Motivated by the challenges of cycles, parametric uncertainties, and scalability, this paper proposes a sensor placement algorithm that guarantees structural observability for cyclic and acyclic networks with parametric uncertainties. By leveraging a graph-based strategy, the proposed method efficiently addresses the computational complexities of large-scale networks. To demonstrate the algorithm's effectiveness, we apply it to several EPANET benchmark WDNs. Most notably, the developed algorithm solves the sensor placement problem with guaranteed structured observability for the L-town WDN with 1694 nodes and 124 cycles in under 0.1 seconds.
\end{abstract}

%-----------------------------------------------------------------------------
\section{Introduction}\label{Sec: Introduction}
Sensor placement is fundamental for state estimation and network monitoring, which is crucial for many real-life large-scale infrastructures, such as water networks. 
While extensively studied, the problem becomes more challenging in large-scale bidirected or undirected cyclic networks with parametric uncertainties, which complicate observability analysis and sensor placement strategies \cite{mohan2023uncertainties}. 
This challenge is particularly relevant in physical systems like water distribution networks (WDNs), which are large-scale, cyclic, and subject to uncertainties in parameters such as pipe lengths, diameters, resistance, and demand patterns. Accurate network monitoring is essential for the fast and precise detection, prediction, and localization of leaks and bursts, as undetected failures can lead to significant water losses, lower quality of service, and costly repairs. This emphasizes the need for scalable sensor placement methods to ensure reliable monitoring and control of WDNs \cite{annaswamy2024control,creaco2019real,perez2014leak}.

Existing sensor placement methods in WDNs rely on simulation-based approaches, where sensor configurations are designed based on simulated leakage scenarios to find the best sensor configuration \cite{santos2022pressure,ROMEROBEN202254,irofti2023learning,gautam2022efficient}. These methods specifically target leakage detection and have strong practical relevance, as they develop sensor placement algorithms based on real-world scenarios. More recent approaches incorporate graph-based and centrality-based heuristics \cite{cheng2023optimal,diao2023sensor,li2023optimal}, which offer improved scalability. Nonetheless, these methods do not come with guarantees of observability, and expected performance is limited to the finite scenarios explored in simulation.
Thus, incorporating control theory, particularly observability theory, could offer a more structured approach for sensor placement, enabling state reconstruction in finite time and thereby improving general network monitoring.

In the control systems literature, various strategies for sensor placement based on observability theory have been proposed. 
In \cite{geelen2021optimal,bopardikar2021randomized}, two approaches leveraging the observability Gramian are introduced, aiming to enhance detectability by maximizing the minimum eigenvalue of the Gramian. The approach in \cite{montanari2020observability} extends this by analyzing the ratio between the minimum and maximum eigenvalues of the Gramian. Nonetheless, these methods rely on the assumption that model parameters, such as pipe resistance and demand patterns, are precisely known.
This limitation is addressed in \cite{zhang2024functional}, where a structural functional observability (SFO) framework is introduced, which accounts for a set of parametric uncertainties and provides necessary and sufficient conditions for observability. This method employs a greedy heuristic to select a sensor placement configuration from a predefined set. Similarly, in \cite{mousavi2020Structural}, structural observability theory is applied to sensor placement in a directed cyclic graph, specifically a motorway Ring Road, where zero forcing sets are used to determine the minimum number of sensors required.

In addition to system theory-based approaches, methods for sensor placement focused on specific network structures have been explored as well.  
In \cite{Suresh2013detection}, a weighted set cover (WSC) algorithm is proposed to place mobile sensors for detecting and localizing events (such as leaks) in acyclic flow networks. In \cite{doostmohammadian2017sensor}, the sensor placement problem is formulated with a focus on cyclic graphs, guaranteeing structural observability. By formulating the sensor placement problem as a Linear Sum Assignment Problem (LSAP). Nonetheless, the method assumes a predefined set of cyclic network structures. 

Motivated by the challenges and limitations of existing sensor placement methods identified above, this paper leverages structural observability principles derived from structural controllability theory \cite{jia2020unifying} to develop a new, scalable, graph-based algorithm for sensor placement in cyclic networks. By relying on structural properties, the method avoids the need for precise parameter knowledge and is applicable to a class of diffusively coupled systems whose dynamics can be approximated by linear ordinary differential equations, and where the resulting structure of the state matrix is symmetric.

To enhance scalability, we develop a method for eliminating cycles and generating spanning trees to determine sensor placement, which significantly reduces computational complexity. 
Here, scalability refers to the ability of the algorithm’s runtime to grow with network size, ideally increasing linearly, i.e., linear computational complexity.
This is essential for real-world WDNs, which often contain thousands of nodes.

In addition, we derive novel sufficient conditions for sensor placement in tree networks that guarantee structural observability. The effectiveness of the proposed algorithm is demonstrated on five EPANET benchmarks of varying size and complexity.

\paragraph*{Novelty with respect to \cite{vanGemert2024exploiting}} 
In \cite{vanGemert2024exploiting}, we proposed a structural observability-based sensor placement method for WDNs, focusing on uncertain and nonlinear WDN models, without considering scalability. Indeed, the sensor placement algorithm in \cite{vanGemert2024exploiting} scales exponentially with the network, requiring 12 hours to find a sensor placement configuration even for small networks (like the Hanoi one).
Therefore, in this paper, we focus on scalability by introducing a spanning tree-based approach that significantly reduces computational complexity while ensuring structural observability in large-scale cyclic networks.

The remainder of this paper is structured as follows. Section \ref{Sec: Preliminaries and problem statement} provides the necessary preliminaries, introducing the considered class of systems, key notions from observability theory, structural system theory, graph theory, and the problem statement. Section \ref{Sec: Main results} presents the main theoretical results and corresponding sensor placement algorithms. Section \ref{Sec: Implementation on EPANET benchmarks} showcases the effectiveness of the proposed method through simulation results, including applications to large-scale water distribution networks, followed by the conclusions in Section \ref{Sec: Conclusions}.

\subsection{Basic notation}
Let $\setR$, $\nonnegR$, and $\posR$ denote the field of real numbers, the set of non-negative reals, and the set of positive reals, respectively. A variable $a\in\{0,1\}$ is called a binary variable. 
For a vector $x\in \setR^{n_x}$, $x_i$ denotes the $i$-th element of $x$. 
For a matrix $A\in \setR^{n_x\times n_x}$, $A^\top$ denotes its transpose and $A(i,:)$ denotes the $i$-th row of $A$. 

%-----------------------------------------------------------------------------
\section{Preliminaries and problem statement}\label{Sec: Preliminaries and problem statement}
In this section, we introduce the necessary preliminaries, as well as the problem statement. 

Let us consider the continuous-time, autonomous, time-invariant, linear system 
\begin{equation}\label{eq: LTI}
\begin{aligned}
    \Dot{x}(t)=&Ax(t), \quad t\in \nonnegR,\\
     y(t) =& Cx(t),
    \end{aligned}
\end{equation}
where $x(t)\in \setX\subseteq \setRnx$ and $y(t)\in\setRny$ are the state and the output of the system, respectively, $A\in\setR^{n_x\times n_x}$ is the state matrix and $C\in\setR^{n_y\times n_x}$ is the output matrix. Assuming these matrices are uncertain, the goal of this paper is to determine a sensor placement combination such that \eqref{eq: LTI} is observable according to the following definition.
\begin{definition}[\hspace{-0.005cm}\cite{kalman1960general}]\label{Def: Observability}
    System \eqref{eq: LTI} is said to be observable if, for any unknown initial state $x(0)\in\setX$, there exists a finite time $t> 0$, such that the output $y(\cdot):[0,t]\rightarrow\setRny$ suffices to uniquely determine $x(0)$.
\end{definition}
To cope with the fact that $A$ and $C$ in \eqref{eq: LTI} are uncertain, we use \emph{structural observability} theory to analyze the observability of \eqref{eq: LTI}. We thus analyze observability by using pattern matrices, which represent the system's structure without relying on the specific numerical values of the elements of $A$ and $C$. To formalize this, we first introduce the concept of a pattern matrix $\graphX$ and its corresponding pattern class $\graphP(\graphX)$.
\begin{definition}\label{Def: PatternMatrix}
A pattern matrix $\graphX \in \strucset^{p\times q}$ is a matrix whose entries belong to the set $\{0, *, ?\}$. The pattern class $\graphP(\graphX)$ is the set of all matrices that %follow the 
have the same structure as the pattern matrix $\graphX$, formally given by:
\begin{align*}
    \graphP(\graphX):=\{X\in\setR^{p\times q}\;|\; &\graphX(i,j)=0\text{ if } X(i,j)=0,\\  &\graphX(i,j)=\ast  \text{ if } X(i,j)\neq 0,\\ &\graphX(i,j)=?\text{ if } X(i,j) \text{ arbitrary} \;\}.
\end{align*}
\end{definition}

This means that $X\in\graphP(\graphX)$ if $X$ follows the structural constraints imposed by $\graphX$ as in Definition \ref{Def: PatternMatrix}.

Using this framework, we define the pattern matrices $\graphA\!\in\!\{0,\ast,?\}^{n_x\times n_x}$ and $\graphC\!\in\!\{0,\ast\}^{n_y\times n_x}$ such that $\MatrinStruc{A}$ and $\MatrinStruc{C}$ represent the family of systems \eqref{eq: LTI}. We refer to this family of systems as a \emph{structured system}, denoted by $(\graphA,\graphC)$. 
The notation of a structured system allows for an extension of the concept from Definition \ref{Def: Observability} to structural observability formalized as follows.
\begin{definition}\label{Def: Structural Obserability}
    A structured system $(\graphA,\graphC)$ is called strongly structurally observable if the pair $(A,C)$ is observable for all $\MatrinStruc{A}$ and $\MatrinStruc{C}$.
\end{definition}
Thanks to the duality between controllability and observability \cite{lee1967foundations}, necessary and sufficient conditions for strong structural observability follow from the conditions devised for strong structural controllability in \cite{jia2020unifying}, by applying the conditions to $(\graphA^\top, \graphC^\top)$.

%-------------------------------------------------------
\subsection{Basic graph notions}
To evaluate strong structural observability and optimize sensor placement in networks, we employ graph-based methods. To facilitate this, we first introduce relevant graph notions.

A graph is defined as $\graphG=\{\graphV,\graphE,\graphW\}$, where $\graphV=\{v_i\}_{i=1}^{n}$ is the set of nodes, $\graphE=\{e_{i}\}_{i=1}^{m}$ is the set of edges that connect nodes and $\graphW=\{w_{i}\}_{i=1}^{m}$ is the set of weights assigned to edges. A graph is \emph{bidirected} if every node with an outgoing edge to another node also has an incoming edge from that same node, though the edge weights may differ. A graph is \emph{undirected} if it is bidirected with equal edge weights. 

For both bidirected and undirected graphs, the adjacency matrix $\Aadj$ represents the edges between pairs of nodes, with its entries representing the weights of the edges. In a bidirected graph, if $\Aadj(i,j) \neq 0$, then $\Aadj(j,i) \neq 0$ as well. However, $\Aadj$ is not symmetric, i.e., $\Aadj(j,i) \neq \Aadj(i,j)$, due to potential asymmetries in the edge weights. For undirected graphs, the adjacency matrix is symmetric, i.e., $\Aadj(j,i) = \Aadj(i,j)$ for all pairs of connected nodes.

The incidence matrix $\Ainc$, represents the relationship between nodes and edges for bidirected graphs, where $\Ainc(i,j) = w_i$ for some $w_i\in\posR$ if $v_i$ is the tail of edge $e_j$ and $\Ainc(i,j) = -w_i$ if $v_i$ is the head of edge $e_j$. 

As an example, consider the graph in Fig. \ref{Fig:Example Networks}(a), the corresponding adjacency and incidence matrices are then given by
\begin{equation}\label{eq:IncAdjWdn}
   \Aadj\!=\!\begin{bmatrix}
         0\!\!&\!\! 1\!\!&\!\! 1 \!\\
         1\!\!&\!\! 0\!\!&\!\! 1\!\\
         1\!\!&\!\! 1\!\!&\!\! 0\!\\
    \end{bmatrix}\!\!, \; \Ainc\!=\!\begin{bmatrix}
         1\!\!&\!\!-1\!\!&\!\! 0\!\!&\!\! 0 \!\!&\!\!-1\!\!&\!\! 1\\
        -1\!\!&\!\! 1\!\!&\!\! 1\!\!&\!\!-1 \!\!&\!\! 
         0\!\!&\!\! 0\\
         0\!\!&\!\! 0\!\!&\!\!-1\!\!&\!\! 1 \!\!&\!\! 1\!\!&\!\!-1\\

    \end{bmatrix}\!\!.
\end{equation}
By relying on the previous notions, the system in \eqref{eq: LTI} can be associated with a graph $\graphG(A)=\{\graphV,\graphE,\graphW\}$, where the state matrix serves as the adjacency matrix.

\subsection{Instrumental graph definitions}
In this section, we introduce some key graph-theoretic definitions that are used throughout this paper, beginning with the concept of a \emph{cycle}, an \emph{inner cycle} and a \emph{cyclic graph}.
\begin{definition}\label{Def: Cycle_InnerCycle_Cyclic}
A \emph{cycle} is a path consisting of more than two nodes, that begins and ends at the same node, with no repeated edges.
An \emph{inner cycle} is a cycle not containing any smaller cycles within it.
A \emph{cyclic graph} is a graph that contains at least one cycle.
\end{definition}
In contrast, a \emph{tree} is a connected graph that is acyclic and thus contains no cycles \cite{diestel2025graph}. A closely related concept is that of a \emph{spanning tree} which is formalized as follows.
\begin{definition}
Given a cyclic graph $\graphG(A) = \{\graphV, \graphE,\graphW\}$, a spanning tree is a subgraph of $\graphG(A)$ denoted by $\graphG(A_{\text{tree}})$, where 
\begin{equation}
\graphG(A_{\text{tree}}) = (\graphV, \graphE_{\text{tree}}, \graphW_{\text{tree}})
\end{equation}
with $\graphE_{\text{tree}} \subset \graphE$ and $\graphW_{\text{tree}} \subset \graphW$ such that $\graphG(A_{\text{tree}})$ is a tree.
\end{definition}
For an example of a spanning tree, we refer to the cyclic bidirected network in Fig. \ref{Fig:Example Networks}(a), with two possible spanning trees shown in Fig. \ref{Fig:Example Networks}(b).

To further classify nodes within a graph, we introduce the concepts of \emph{intersection nodes} and \emph{extreme nodes}.
\begin{definition}\label{Def: Intersection_Extreme_Node}
    A node $v \in \graphV$ is an \textit{intersection node} if it has at least three neighbors, with the set of intersection nodes with cardinality $n_i$ denoted by $\graphVint \subset \graphV$. A node $v \in \graphV$ is an \textit{extreme node} if it has exactly one neighbor, with the set of extreme nodes with the cardinality $n_e$ denoted by $\graphVext \subset \graphV$. 
\end{definition}
Finally, we introduce the definitions of a \emph{path}, \emph{branch} and \emph{self-loop}. 
\begin{definition}\label{Def: Branch}
    A \emph{path} is a sequence of nodes $v_i \in \graphV$ connected by edges $e_i \in \graphE$, namely $(v_1, v_2, \dots, v_k)$ for some $k\leq n$.
    A \emph{branch} is a path that starts at an \emph{extreme node} $v_1$ and continues, $(v_1, v_2, \dots, v_k)$ until it reaches an \emph{intersection node} $v_k$, i.e. $(v_1, v_2,\ldots,v_k)$, where $v_i$ for $i=2,\hdots,k-1$, are neither extreme nor intersection nodes.
    A \emph{self-loop} is an edge that connects a node to itself, i.e., for a node $v \in \graphV$, there exists an edge $e = (v, v) \in \graphE$.\end{definition}

\subsection{Colorability rules for graphs}
To verify the strong structural observability of a structured system, we use a graph-based approach that applies a color-change rule for observability.
Before introducing the color-change rule, we define an unweighted graph $\graphG(\graphX) = \{\graphV, \graphE\}$ associated with a pattern matrix $\graphX$, where $\graphX$ serves as the adjacency matrix. Specifically, an edge exists from node $v_i$ to node $v_j$ if $\graphX(i,j) = \ast$ or $\graphX(i,j) = ?$.
To distinguish between the entries \( \ast \) and \( ? \) we first introduce two subsets $\graphE_\ast$ and $\graphE_?$, of the edge set \( \graphE \), with   
\begin{align*}
    (i, j) \in \graphE_\ast \;\; &\text{iff} \;\; \graphX(i,j) = \ast, \quad
    (i, j) \in \graphE_? \;\; \text{iff} \;\; \graphX(i,j) = ?.
\end{align*}
The edges in $\graphE_\ast$ and $\graphE_?$ are represented by solid and dashed arrows, respectively. As an example, consider the pattern matrix
\begin{equation}\label{eq: PatternMatrix}
     \graphX\!=\!\begin{bmatrix}
         0\!\!&\!\! \ast\!\!&\!\! \ast \!\\
         \ast\!\!&\!\! 0\!\!&\!\! ?\!\\
         \ast\!\!&\!\! ?\!\!&\!\! 0\!\\
    \end{bmatrix}.  
\end{equation} 
The associated graph $\graphG(\graphX)$ is shown in Fig. \ref{Fig:Example Networks}(c).
\begin{figure}
    \centering
    \includegraphics[width=0.8\linewidth]{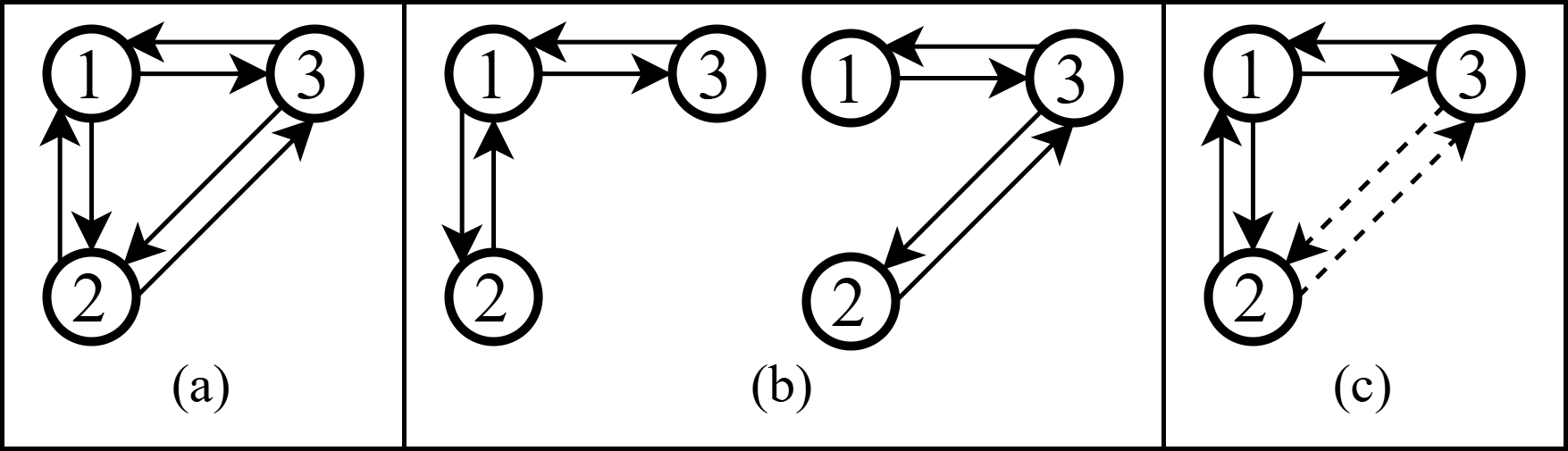}
    \caption{Subfigure (a) illustrates an example of a bidirected cyclic network with the structure of $\Aadj$ in \eqref{eq:IncAdjWdn}. Subfigure (b) presents two possible spanning trees derived from (a), and subfigure (c) displays a graph with arbitrary edges.}
    \label{Fig:Example Networks}
\end{figure}

Next, we introduce the \emph{\textbf{color-change rule}} based on the graph $\graphAC$.
Every node in the graph $\graphAC=\{\graphV,\graphE_\ast\}$ is initially colored white. If a node \( i \) has exactly one out-neighbor \( j \) such that \( (i,j) \in \graphE_\ast \) and \( j \) is white (self-loops included), then we change the color of node \( j \) to black, see Fig. \ref{Fig:Lemma_Example}. The graph \( \graphAC \) is considered \emph{colorable} if all nodes \( \graphV(i) \) for \( i = 1, \dots, n_x \), representing the nodes of \(\graphG(\graphA^\top) \), are colored.

With the color-change rule, we can define the condition that leads to \emph{strong structural observability} in Definition \ref{Def: Structural Obserability}.
\begin{definition}[\hspace{-0.005cm}\cite{jia2020unifying}]\label{Def: Observability_by_Colorability}
     A structured system $(\graphA,\graphC)$ is strongly structurally observable if the graphs $\graphAC$ and $\graphACbar$, where  \begin{equation}
    \graphAbar(i,i):= \begin{cases}\ast & \text { if } \graphA(i,i)=0  \\ ? & \text { otherwise,}\end{cases}
\end{equation} are colorable. 
\end{definition}
Note that constructing $\graphAbar$ from $\graphA$ affects only the self-loops in $\graphACbar$ by modifying the diagonal elements of $\graphAbar$, i.e., no $0$ elements on the diagonal.

\subsection{Problem statement}\label{SubSec: Assumptions and goal}
Assume we have system \eqref{eq: LTI} with $A\in\graphP(\graphA)$ and $C\in\graphP(\graphC)$ where $\graphA\in\{0,\ast,?\}^{n_x\times n_x}$ and $\graphC\in\{0,\ast\}^{n_y\times n_x}$, satisfy the following assumptions. 
\begin{assumption}\label{Ass: GraphA}
The graph structure $\graphG(\graphA^\top)$ satisfies:
\begin{enumerate}
    \item[i)] The graph $\graphG(\graphA^\top)$ is bidirected or undirected, i.e, the matrix \(\graphA\) is symmetric such that \( \graphA = \graphA^\top \);
    \item[ii)] The graph $\graphG(\graphA^\top)$ is \emph{fully connected}, with all pairs of nodes having a path between them through edges in $\graphE_\ast$;
    \item[iii)] The graph $\graphG(\graphA^\top)$ contains \emph{at least one extreme node}. 
\end{enumerate}    
\end{assumption}
\begin{assumption}\label{Ass: C}
    The output matrix $C\in\setR^{n_y\times n_x}$ satisfies:
    \begin{itemize}
        \item[i)] $C(i,j)\in\{1,0\}$ for all $i=1,\ldots,n_y$ and $j=1,\ldots,n_x$;
        \item[ii)] $\sum_{j=1}^{n_x}C(i,j)=1$ and $\sum_{i=1}^{n_y}C(i,j)\leq1$;
        \item [iii)] rank$(C)=n_y$.
    \end{itemize}
\end{assumption}
These assumptions on the output matrix ensure each row contains exactly one entry equal to $1$, and the rows are linearly independent. Note that this structure reflects a realistic WDN scenario, where sensors, such as pressure and flow sensors, are limited to measuring a single state, e.g., a flow sensor can only measure the flow rate in one specific pipe. 

Now, suppose we attempt to evaluate all possible sensor combinations to retrieve the optimal sensor placement configuration. As the network size increases, the computational complexity of this problem grows exponentially. Specifically, if Assumption \ref{Ass: C} holds, then, for $A \in \setR^{n_x \times n_x}$, there are $2^{n_x} - 1$ possible sensor configurations (and thus possible configurations of $C$). Consequently, the size of this problem explodes, making it nontrivial to develop an algorithm that efficiently finds a sensor placement combination for large networks.

Motivated by this, under Assumption \ref{Ass: GraphA} and Assumption \ref{Ass: C} our objective is to develop a scalable sensor placement strategy that ensures observability of system \eqref{eq: LTI}, specifically targeting large-scale bidirected or undirected cyclic networks with parametric uncertainties.

%-----------------------------------------------------------------------------
\section{Main results}\label{Sec: Main results}
We first present an intermediate result on sensor placement in bidirected or undirected tree graphs, as their simpler structure enables more straightforward sensor placement. We then focus on our main objective, a scalable sensor placement algorithm for cyclic bidirected or undirected graphs, which uses spanning trees to reduce computational complexity.

\subsection{Sensor placement in tree graphs using colorability}\label{SubSec: Sensor placement in tree graphs using colorability}
Let a pattern matrix $\graphA$ for \eqref{eq: LTI}, such that $A\in \graphP(\graphA)$, satisfy the following assumption.
\begin{assumption}\label{Ass: ATree}
    The graph $\graphG(\graphA^\top)$ is a tree.
\end{assumption}
Note that if $\graphG(\graphA^\top)$ is a tree, then $\graphG(\graphAbar^\top)$ is also a tree, as self-loops do not create cycles.
Under this assumption, we derive the following result on the observability of \eqref{eq: LTI}.
%LEMMA----------------------------------------------------------------------------------------------
\begin{lemma}\label{Lemma: Tree}
Let Assumption \ref{Ass: C} and Assumption \ref{Ass: ATree} hold, and let $C\in\setR^{n_y\times n_x}$ satisfy    
\begin{equation}\label{eq: C_Lemma}
            C(i,j)=\begin{cases}1 &\text{for } j\in\graphVext\\0 &\text{otherwise},\\ \end{cases}
    \end{equation} with $n_y = n_e - 1$ and $n_e,\graphV_e$ introduced in Definition \ref{Def: Intersection_Extreme_Node}. Then, the linear system \eqref{eq: LTI} is observable. 
\end{lemma}
\begin{proof}
Let $\graphC$ be a pattern matrix such that $C \in \graphP(\graphC)$. If Assumption \ref{Ass: C} holds and $C$ satisfies \eqref{eq: C_Lemma}, then, by construction, a sensor node is connected to every extreme node of $\graphA^\top$ and $\graphAbar^\top$ except for one, in $\graphAC$ and $\graphACbar$, respectively.
By following the coloring rule and considering only edges in $\graphE_\ast$ for the coloring process, we can conclude the following on $\graphAC$ and $\graphACbar$:
\begin{enumerate}
    \item All \emph{extreme nodes} except one are colored by the sensor nodes connected to them, due to the construction of $C$.
    \item Given the tree structure of $\graphA^\top$ and $\graphAbar^\top$, each branch with a colored extreme node is then colored.  
    \item The remaining uncolored nodes are either part of the single uncolored branch or lie between intersection nodes that are already colored. Therefore, due to the tree structure, the nodes between the intersection nodes can also be colored following the color rule.
    \item Then, the remaining uncolored branch is colored.
\end{enumerate}
Thus, both $\graphAC$ and $\graphACbar$ are colorable. As both $\graphAC$ and $\graphACbar$ are colorable, $(\graphA, \graphC)$ is \emph{strongly structurally observable} by Definition \ref{Def: Observability_by_Colorability}, and thus, system \eqref{eq: LTI} is \emph{observable} by Definition \ref{Def: Structural Obserability}, completing the proof.
\end{proof}

To illustrate the coloring steps in the proof of Lemma \ref{Lemma: Tree}, we introduce the following example.
\begin{example}\label{Ex: Tree}
    Consider the tree graph $\graphAC$ in Fig. \ref{Fig:Lemma_Example}(a). The graph $\graphG(\graphA^\top)$ is a tree with its nodes represented by gray circles in Fig \ref{Fig:Lemma_Example}, while the hexagon nodes represent the nodes from $\graphC^\top$ in $
    \graphG(\graphA^\top,\graphC^\top)$. The extreme nodes are $\graphV_e = {1, 3, 7}$ and the sensor nodes are $10$ and $11$. The sensor matrix $C$ is constructed according to Assumption \ref{Ass: C} and \eqref{eq: C_Lemma} in Lemma \ref{Lemma: Tree}. Following the steps in the proof of Lemma \ref{Lemma: Tree}, nodes 1 and 7 are colored by the sensor nodes (Fig. \ref{Fig:Lemma_Example}(a)). As extreme nodes, they color their connected branches and the intersection node $5$ (Fig. \ref{Fig:Lemma_Example}(b-e)). As the graph is a tree, all uncolored nodes between colored intersection nodes are colored, and finally, the remaining uncolored branch is also colored (Fig. \ref{Fig:Lemma_Example}(f)). Thus, $\graphAC$ is colored.
\end{example}
 \begin{figure}[ht]
    \centering
    \includegraphics[width=0.8\linewidth]{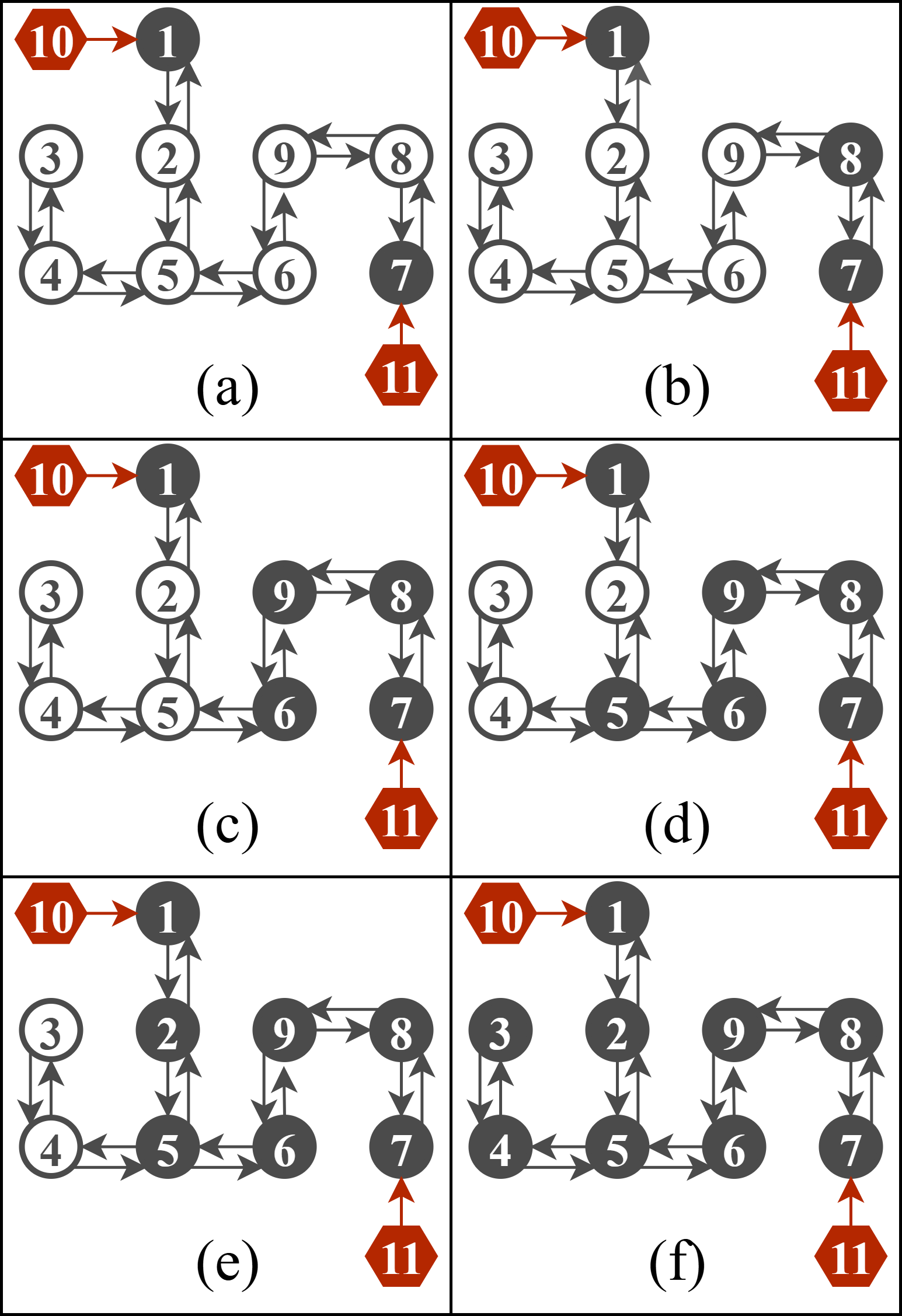}\vspace{0.1cm}
    \caption{Subfigures (a)-(f) illustrate the sequential application of the color change rule on the graph $\graphG(\graphA^\top,\graphC^\top)$. The round nodes represent the nodes associated with $\graphA^\top\!$, while the red hexagon nodes indicate the sensor nodes corresponding to $\graphC^\top\!$.}
    \label{Fig:Lemma_Example}
\end{figure}
\begin{remark}\label{Remark: self-loops}
The coloring process for $\graphACbar$ remains unaffected, as self-loops do not influence the steps involved. Namely, in Example \ref{Ex: Tree}, the nodes $1, 2, \dots, 9$ in $\graphACbar$ have self-loops. In the first step of the proof of Lemma \ref{Lemma: Tree}, the extreme nodes are colored by the sensor nodes. If an extreme node is not colored, it would have two white out-neighbors in $\graphE_\ast$, one of which is the self-loop. However, since the coloring starts from the sensor nodes, the process ensures that $\graphACbar$ remains colorable by following the same steps.
\end{remark}

\subsection{Sensor placement in cyclic graphs using spanning trees and colorability}\label{Subsec: Sensor placement in cyclic graphs using spanning trees and colorability}

We now introduce a scalable sensor placement algorithm for cyclic bidirected or undirected graphs. As shown in Fig. \ref{fig:Flowchart}, the algorithm follows four steps: i) start with a cyclic graph, ii) generate an associated spanning tree, iii) determine a sensor placement strategy based on the spanning tree, and iv) establish colorability of the cyclic graph under the determined sensor configuration.

\begin{figure}[b]
    \centering
    \includegraphics[width=0.95\linewidth]{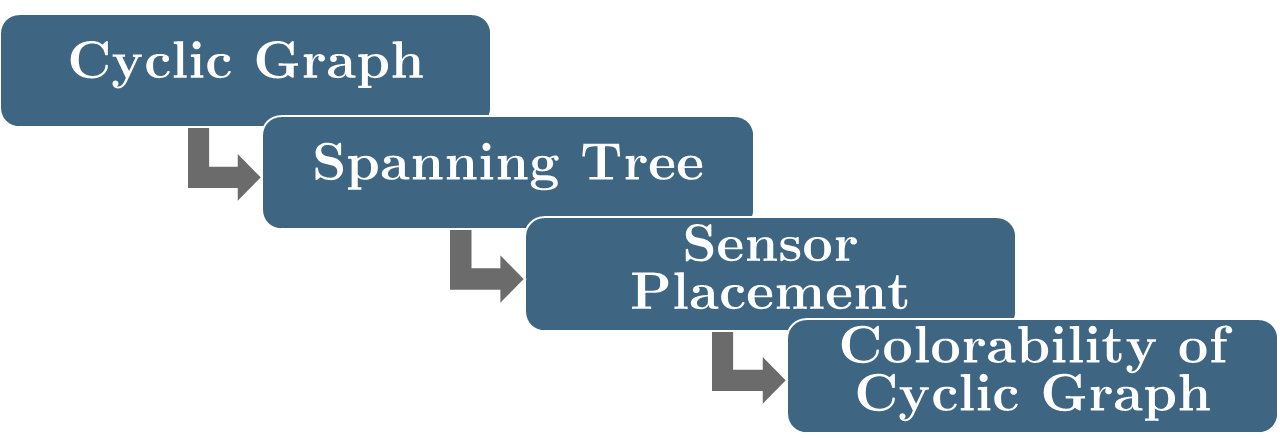}
    \caption{Building blocks for developing a scalable sensor placement algorithm.}
    \label{fig:Flowchart}
\end{figure}

First, we introduce Algorithm \ref{Alg:convert_to_tree}, which extracts a spanning tree from a bidirected or undirected cyclic graph $\graphG(\graphA)$. The input is the adjacency matrix $\graphA \in \{0, \ast, ?\}^{n_x \times n_x}$, where the symbol $\ast$ in Algorithm \ref{Alg:convert_to_tree} is implemented using a symbolic variable \texttt{star}. The output is the tree adjacency matrix $\graphA_{\text{tree}}$. 
The algorithm initializes the visited node vector $\textit{v}$, parent tracking vector $\textit{p}$, and tree adjacency matrix $\graphA_{\text{tree}}$ (lines 3–4). A depth-first search (DFS) explores $\graphG(\graphA)$, adding edges to $\graphA_{\text{tree}}$ only when encountering an unvisited node (lines 10-11), thereby ensuring $\graphG(\graphAtree)$ is a tree. Parent nodes are recorded (line 11), and DFS recurses until all nodes are processed. The resulting $\graphA_{\text{tree}}$ has a cycle-free spanning structure, specifically eliminating edges between intersection nodes and one of their neighboring nodes in a cycle.
\begin{algorithm}[t]
    \caption{Cyclic Graph to Spanning Tree using DFS}
    \begin{algorithmic}[1]\label{Alg:convert_to_tree}
    \STATE \textbf{Input:} Adjacency matrix $\graphA$ of size $n_x\times n_x$
    \STATE \textbf{Output:} Tree adjacency matrix $\graphA_{\text{tree}}$ of size $n_x\times n_x$
    \STATE $\textit{v} = \textit{false}(1,n_x)$, 
    \STATE $\textit{p} = -1 \times \textit{ones}(1,n_x)$, $\graphA_{\text{tree}} = \textit{zeros}(n_x,n_x)$ 
    \FOR{$i = 1,\dots,n_x$}
        \IF{$\textit{v}(i) = \textit{false}$} 
            \STATE $\textit{v}(i) = \textit{true}$
            \FOR{$j = 1,\dots,n_x$} 
                \IF{$\graphA (i,j) = \ast$ \textit{and} $\textit{v}(j) = \textit{false}$} 
                    \STATE $\graphA_{\text{tree}}(i,j) = \ast$, $\graphA_{\text{tree}}(j,i) = \ast$ 
                    \STATE $\textit{p}(j) = i$, $\textit{v}(j) = \textit{true}$
                    \FOR{$k = 1,\dots,n_x$} 
                        \IF{$\graphA (j,k) = \ast$ \textit{and} $\textit{v}(k) = \textit{false}$} 
                            \STATE $\graphA_{\text{tree}}(j,k) = \ast$, $\graphA_{\text{tree}}(k,j) = \ast$ 
                            \STATE $\textit{p}(k) = j$
                        \ENDIF
                    \ENDFOR
                \ENDIF
            \ENDFOR
        \ENDIF
    \ENDFOR
    \end{algorithmic}
\end{algorithm}

To evaluate the computational complexity of Algorithm \ref{Alg:convert_to_tree}, we use O-notation to describe its upper bound \cite[Chapter 3]{cormen2022introduction}. The complexity is determined by the for loops in the algorithm: The first loop iterates over all nodes, contributing $\text{O}(n_x)$. The second and third loops iterate over subsets of nodes based on the graph structure, contributing $\text{O}(n_x - c_1)$ and $\text{O}(n_x - c_2)$, respectively, where $c_1$ and $c_2$ are integers bounded between 0 and $n_x$. Therefore, the total complexity is $\text{O}(n_x(n_x - c_1)(n_x - c_2))$. In the worst case, when the graph is fully connected, the complexity becomes $\text{O}(n_x^3)$, while in the best case, it reduces to $\text{O}(n_x)$.

After obtaining the spanning tree using Algorithm \ref{Alg:convert_to_tree}, we proceed to place sensor nodes based on the spanning tree $\graphAtree$, see Algorithm \ref{Alg:sensor_placement}. Algorithm \ref{Alg:sensor_placement} initializes the sensor placement process by first removing self-loops from the tree adjacency matrix (line 3). It then identifies the extreme nodes of $\graphAtree$ (line 4), and places sensor nodes connected to these nodes by updating the sensor placement matrix $C$ (lines 6-8). The computational complexity of Algorithm \ref{Alg:sensor_placement} is $\text{O}(n_x)$, determined by the for loop in lines 6–8.

\begin{algorithm}[t]
    \caption{Sensor Placement for Cyclic Graphs Based Spanning Tree}
    \begin{algorithmic}[1]\label{Alg:sensor_placement}
    \STATE \textbf{Input:} Adjacency matrix $\graphA$ of size $n_x \times n_x$, Tree adjacency matrix $\graphA_{\text{tree}}$ of size $n_x \times n_x$
    \STATE \textbf{Output:} Sensor placement matrix $C$ of size $n_e \times n_x$, 
    \STATE $\graphA_{\text{tree}}(i,i) = 0$
    \STATE $\graphV_e=find(\sum(\graphA_{\text{tree}}(i,:) < 2)$
    \STATE $n_e=length(\graphV_e)$, $C = \text{zeros}(n_e, n_x)$
    \FOR{$i = 1 : n_e$}
            \STATE $C(i,\graphV_e(i)) = 1$ 
    \ENDFOR
    \end{algorithmic}
\end{algorithm}

Based on the above Algorithms \ref{Alg:convert_to_tree} and \ref{Alg:sensor_placement}, the main result is stated next.

%THEOREM-----------------------------------------------------------------------------------------------------------------
\begin{theorem}\label{Theorem:Colorability_cyclic}
    Let Assumptions \ref{Ass: GraphA} and \ref{Ass: C} hold. Let $\graphAtree$ be the spanning tree of $\graphA$ generated by Algorithm \ref{Alg:convert_to_tree}, and let $C$ be the sensor matrix constructed by Algorithm \ref{Alg:sensor_placement}. Then, the linear system \eqref{eq: LTI} is observable.
\end{theorem}
\begin{proof} 
    Let $\graphC$ be a pattern matrix such that $C\in\graphP(\graphC)$. If Assumption \ref{Ass: C} holds and $C$ is constructed by Algorithm \ref{Alg:sensor_placement}, with $\graphAtree$ generated by Algorithm \ref{Alg:convert_to_tree}, then a sensor node is connected to at least one extreme node and nodes neighboring an intersection node of $\graphA^\top$ and $\graphAbar^\top$ in $\graphAC$ and $\graphACbar$, respectively.
    By following the coloring rule and considering only edges in $\graphE_\ast$ for the coloring process, we can conclude the following on $\graphAC$ and $\graphACbar$:
    \begin{enumerate}
        \item All nodes connected to sensor nodes defined by $C$ in $\graphAC$ and $\graphACbar$ are colored.
        \item Following Definition \ref{Def: Intersection_Extreme_Node} and Definition \ref{Def: Branch}, the colored extreme nodes color the associated branches.
        \item From 1)-2) and due to the construction of $C$, there is \emph{at least} one colored node $i$ with exactly one white out-neighbour $j$ with $(i,j)\in\graphE_\ast$. Thus, following the color rule, this node can color its neighbors until it reaches an intersection node, at which point it can color the intersection node.
        \item Due to the sensor node placement in cycles based on the construction of $C$, step 4 repeats until both $\graphAC$ and $\graphACbar$ are fully colored.
    \end{enumerate}
    Then, as $\graphAC$ and $\graphACbar$ are colorable, $(\graphA,\graphC)$ is \emph{strongly structurally observable} by Definition \ref{Def: Observability_by_Colorability}, and thus, system \eqref{eq: LTI} is \emph{observable} by Definition \ref{Def: Structural Obserability}, completing the proof. 
\end{proof}

To illustrate the coloring steps in the proof of Theorem \ref{Theorem:Colorability_cyclic}, we introduce the following example.
\begin{example}
    Consider the cyclic graph $\graphAC$ shown in Fig. \ref{Fig:Theorem_Example}(a), with its spanning tree constructed by Algorithm \ref{Alg:convert_to_tree} shown in Fig. \ref{Fig:Lemma_Example}(a). The nodes corresponding to the pattern matrix $\graphA$ are shown in grey and the nodes corresponding to the pattern matrix $\graphC$ are depicted in red. By Lemma \ref{Lemma: Tree}, the spanning tree is colorable if sensor nodes are connected to all extreme nodes except one. However, for the cyclic graph, the sensor matrix $C$ must ensure that a sensor node is connected to each extreme node of the spanning tree. Thus, sensor nodes $10$, $11$ and $12$ connected to nodes $1$, $3$ and $7$, respectively. 
    Fig. \ref{Fig:Theorem_Example}(a) shows the sensor nodes connected to at least one extreme node (node $1$) and nodes neighboring an intersection node (nodes $3,7)$, consistent with the proof of Theorem \ref{Theorem:Colorability_cyclic}. 
    
    The coloring process starts by coloring the nodes connected to sensor nodes, as shown in Figure \ref{Fig:Theorem_Example}(a). Next, the branch connected to colored extreme nodes is colored in Figure \ref{Fig:Theorem_Example}(b). In Figure \ref{Fig:Theorem_Example}(c), step 3 of the proof is applied to node $3$, which then colors nodes $4$ and $5$. This step repeats in Figures \ref{Fig:Theorem_Example}(d-f) until $\graphAC$ is fully colored. As stated in Remark \ref{Remark: self-loops}, the same steps apply to $\graphACbar$, ensuring it is also colorable. Thus, $(\graphA, \graphC)$ is strongly structurally observable.
\end{example}
\begin{figure}[t]
    \centering
    \includegraphics[width=0.9\linewidth]{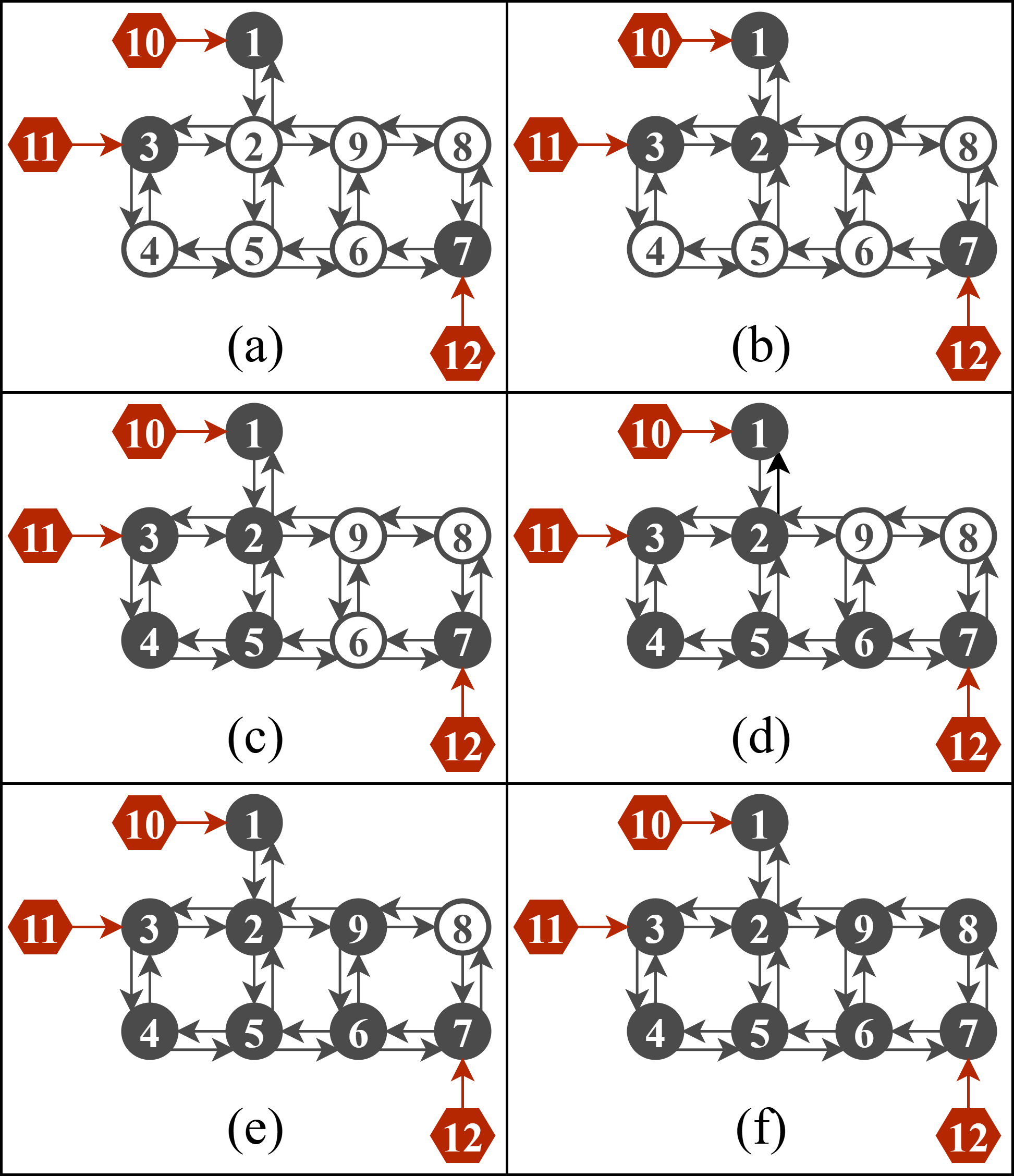}\vspace{0.1cm}
    \caption{Subfigures (a)-(f) illustrate the sequential application of the color change rule for cyclic graphs. The round nodes represent the nodes associated with $\graphA^\top\!$, while the red hexagon nodes indicate the sensor nodes corresponding to $\graphC^\top\!$.} 
    \label{Fig:Theorem_Example}
\end{figure}

\begin{algorithm}
    \caption{Incidence Matrix Generation}
    \begin{algorithmic}[1]\label{Alg:incidence_matrix}
    \STATE \textbf{Input:} EPANET object \texttt{G}.
    \STATE \textbf{Output:} Incidence matrix $\Ainc$ of size $n_j \times n_p$.
    \STATE \texttt{linkTable}$= G.getLinkNodesIndex $
    \STATE $n_l=size(\texttt{linkTable},1)$, $n_p=G.NodeCount$ 
    \STATE $\Ainc=zeros(n_j,n_p)$
    \FOR{$i = 1 : n_p$}
        \STATE $v_{out} = \texttt{linkTable}(i, 1)$  
        \STATE $v_{in} = \texttt{linkTable}(i, 2)$  
        \STATE $\Ainc(v_{out},i)= 1 $, $\Ainc(v_{in},i)= -1 $ 
    \ENDFOR
    \end{algorithmic}
\end{algorithm}

%-----------------------------------------------------------------------------
\section{Implementation on EPANET benchmarks}\label{Sec: Implementation on EPANET benchmarks}

To demonstrate the effectiveness of the proposed sensor placement algorithm, we apply it to the Hanoi, AnyTown, Net3, D-town and L-town networks from the EPANET MATLAB toolbox, representing a range of sizes and configurations, as summarized in Table \ref{tab: Results}. 

For the sensor placement Algorithms \ref{Alg:convert_to_tree} and \ref{Alg:sensor_placement}, a structured representation of a Water Distribution Network (WDN) model is required. Therefore, for all examples, we utilize the linear elastic water column (EWC) model, as described in \cite{zeng2022elastic}, with the associated structured system outlined in \cite[Section IV.A]{vanGemert2024exploiting} to capture the dynamics of the WDNs.
To implement this in MATLAB, particularly for EPANET networks, we define a function to construct the incidence matrix $\Ainc$, which is needed for the structured WDN model. Algorithm \ref{Alg:incidence_matrix} utilizes the EPANET network link data (lines 3-4) to generate the incidence matrix $\Ainc$ (lines 6-10). Using $\Ainc$, we define the structured system dynamics as detailed in \cite[Section IV.A]{vanGemert2024exploiting}, and implement the system in MATLAB using symbolic variables \texttt{star} and \texttt{Question}. The computational complexity of Algorithm \ref{Alg:incidence_matrix} is $\text{O}(n_x)$, determined by the for loop in lines 6–10.

The results of the implemented spanning tree Algorithm \ref{Alg:convert_to_tree} combined with the sensor placement Algorithm are shown in Table \ref{tab: Results}\footnote{All experiments were conducted on a Lenovo ThinkPad with an Intel Core i7-165H vPro, 64GB RAM, running MATLAB R2023a.}. 
As can be seen in Table \ref{tab: Results}, the proposed algorithms compute a sensor placement configuration in under 0.1 seconds, even for the largest network, demonstrating its scalability.
The number of placed sensors depends on the number of extreme nodes, cycles, and network topology. In most cases, e.g., Hanoi, Net3, and D-town, sensors are placed at extreme nodes, with one additional sensor per cycle, resulting in a total count equal to the sum of extreme nodes and cycles. This aligns with the structure of the sensor placement algorithm, which disconnects cycles and places sensors at the resulting extreme nodes, thereby ensuring at least one sensor per cycle. However, in more interconnected networks with overlapping cycles and fewer extreme nodes, such as AnyTown and L-Town, additional sensors may be needed within cycles.

\begin{table}[ht]
\footnotesize
\centering
\begin{tabular}{ |p{1cm}|p{0.6cm}|p{0.7cm}|p{0.85cm}|p{0.8cm}|p{1.6cm}| }
 \hline
 \mbox{Name} & \mbox{Nodes} & \mbox{Cycles} & \mbox{Extreme} Nodes & \mbox{Sensors} & \mbox{Computation} \mbox{Time [s] of} \mbox{Alg. \ref{Alg:convert_to_tree}} {and \ref{Alg:sensor_placement}} \\
 \hline 
 Hanoi   & 66  & 3   & 3   & 6   & 0.0287 \\
 AnyTown & 71  & 19  & 2   & 24  & 0.0302 \\
 Net3    & 216 & 23  & 16  & 39  & 0.0325 \\
 D-town  & 866 & 53  & 78  & 131 & 0.0604 \\
L-town  & 1694 & 124 & 37  & 162 & 0.0993 \\
 \hline
\end{tabular}
\caption{Computation times (in seconds) for the scalable algorithms presented in this paper for various networks.}
\label{tab: Results}
\end{table}
In Figures \ref{Fig: Hanoi_network} and \ref{Fig: D-town_Network}, the Hanoi and D-town networks are visualized, while Figures \ref{Fig: Hanoi} and \ref{Fig: D-town} show the corresponding sensor placements using Algorithms \ref{Alg:convert_to_tree} and \ref{Alg:sensor_placement}. For the Hanoi network, sensor nodes are placed at extreme nodes and near intersection nodes, consistent with the proof of Theorem \ref{Theorem:Colorability_cyclic}. 
The D-town network illustrates the complexity of large-scale cyclic structures and emphasizes the need for scalable sensor placement strategies.

While the proposed method prioritizes scalability and performs well on large-scale benchmarks, it does not consider sensor costs or preferred locations, and the resulting placement is not guaranteed to be unique. In contrast, the approach in \cite{vanGemert2024exploiting} incorporates these factors but does not scale, requiring up to 12 hours even for small networks like Hanoi. Consequently, the sensor configurations obtained here may differ from those in \cite{vanGemert2024exploiting}, particularly when cost plays a significant role.

\begin{figure}[t]
    \centering
    \includegraphics[width=0.5\linewidth]{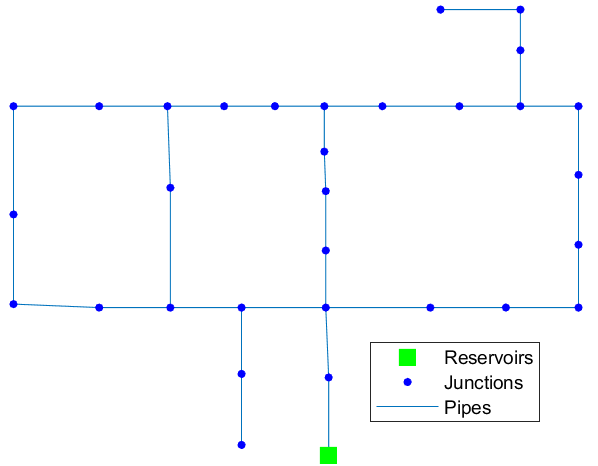}
    \caption{The Hanoi network from EPANET.}
    \label{Fig: Hanoi_network}
\end{figure}
\begin{figure}[b]
    \centering
    \includegraphics[width=0.8\linewidth]{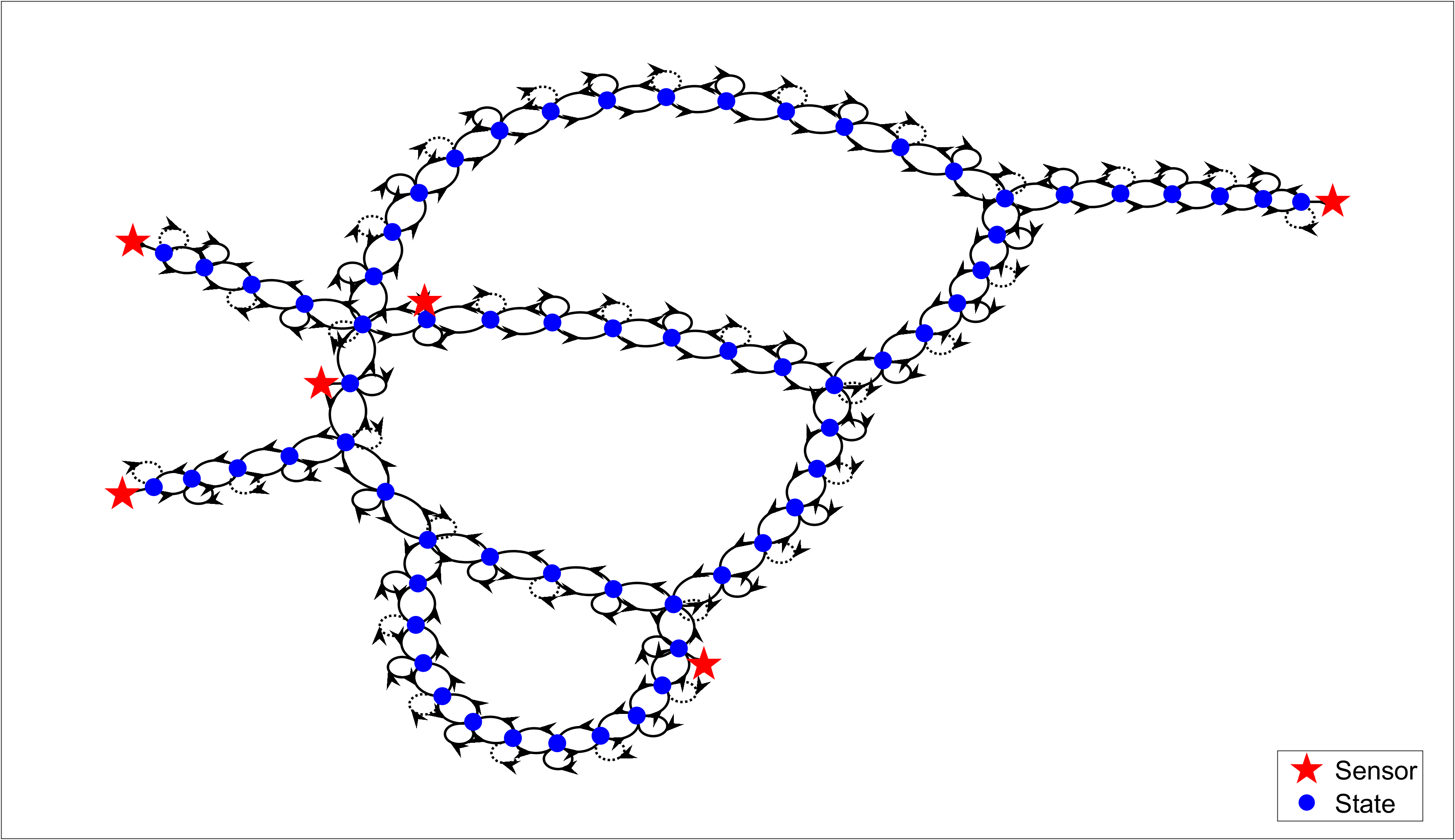}
    \caption{Hanoi network with sensor placement configuration constructed using Algorithm \ref{Alg:sensor_placement}, ensuring observability.}
    \label{Fig: Hanoi}
\end{figure}
\begin{figure}[ht]
    \centering
    \includegraphics[width=0.7\linewidth]{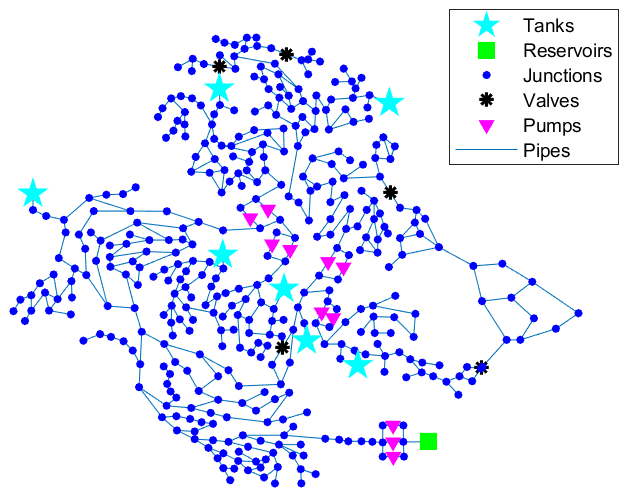}
    \caption{The D-town network from EPANET.}
    \label{Fig: D-town_Network}
\end{figure}
\begin{figure}[ht]
    \centering
    \includegraphics[width=0.99\linewidth]{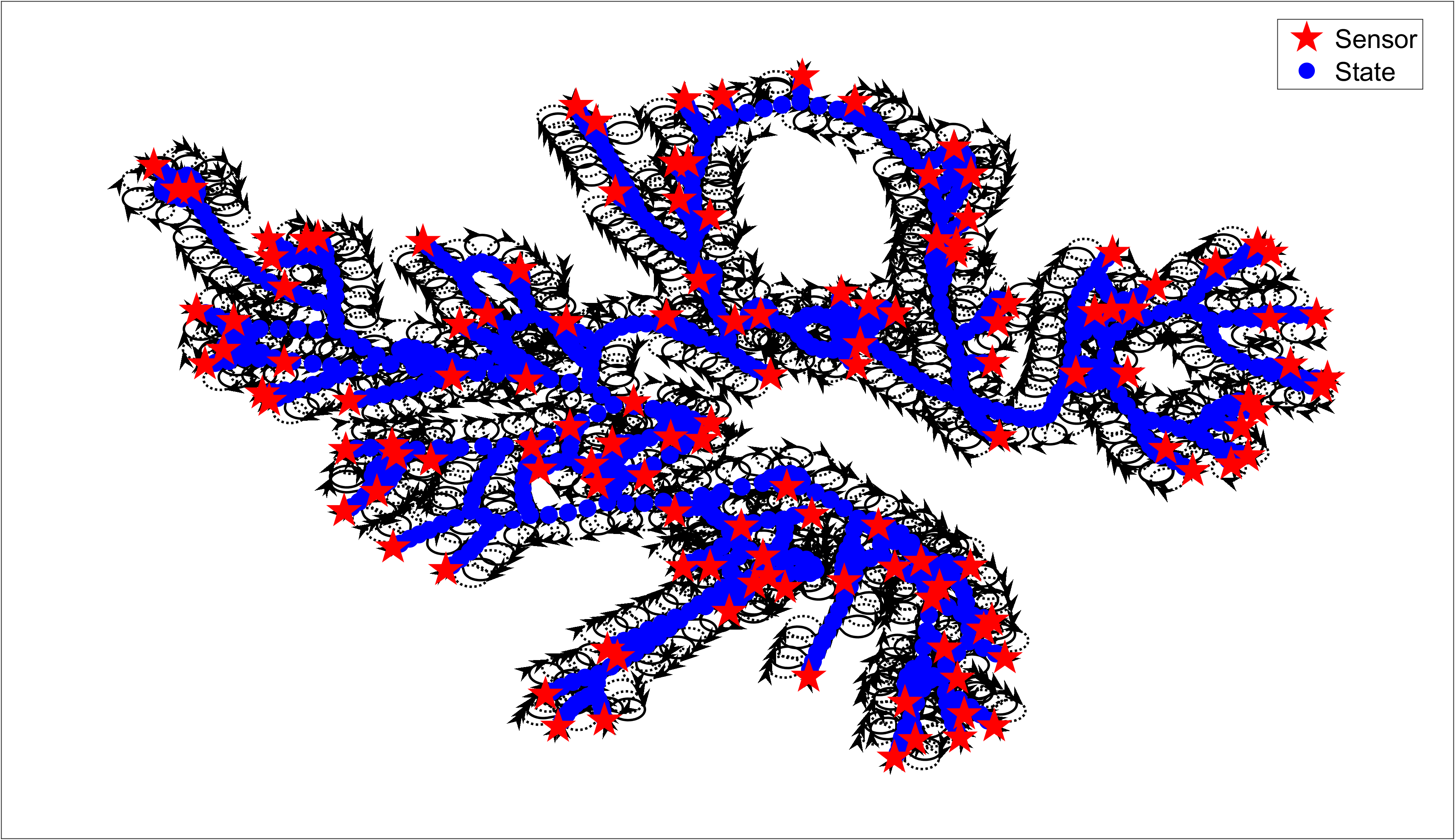}
    \caption{D-town network with sensor placement configuration constructed using Algorithm \ref{Alg:sensor_placement}, ensuring observability}
    \label{Fig: D-town}
\end{figure}

%-----------------------------------------------------------------------------
\section{Conclusions}\label{Sec: Conclusions}
In this paper, we addressed the challenges of optimal sensor placement for state estimation and monitoring in large-scale cyclic and acyclic networks under parametric uncertainties. We proposed a method for sensor placement for acyclic networks and an algorithm for cyclic networks, both ensuring structural observability in the presence of parametric uncertainties. For cyclic networks, a graph-based strategy was employed to efficiently manage the computational complexities associated with large-scale systems. The effectiveness of the proposed approach was demonstrated through its application to several large-scale water distribution networks to highlight its scalability and practical applicability in cyclic network scenarios.

Future work may involve extending the proposed method to accommodate directed networks, which are common in practical applications. To further improve applicability, the sensor placement strategy should be made less conservative and adapted to incorporate sensor costs and preferred installation locations. Additionally, the spanning-tree algorithm could be optimized to achieve linear computational complexity. Finally, a comparative study with existing sensor placement methods for leakage detection in WDNs would offer valuable insights into the effectiveness of observability-based methods in real-world scenarios.
\bibliographystyle{IEEEtran}
\bibliography{main}
\end{document}